\documentclass[runningheads]{llncs}

\usepackage{amsmath}
\usepackage{latexsym}
\usepackage{amsfonts,amssymb}
\usepackage{graphicx, gastex}
\usepackage{cite}
\usepackage{mathrsfs}
\DeclareMathOperator{\dt}{\textbf{.}} 

\DeclareSymbolFont{rsfscript}{OMS}{rsfs}{m}{n}
\DeclareSymbolFontAlphabet{\mathrsfs}{rsfscript}
\newcommand{\Fact}{\textit{Fact}}%{\operatorname{Fact}}
\newcommand{\Pref}{\textit{Pref}}%{\operatorname{Pref}}
\newcommand{\Suff}{\textit{Suff}}%{\operatorname{Suff}}

\begin{document}

\title{On Non-Complete Sets and Restivo's Conjecture}

\author{Vladimir V. Gusev, Elena V. Pribavkina}

\titlerunning{On Non-Complete Sets and Restivo's Conjecture}
\authorrunning{V. Gusev, E. Pribavkina}

\institute{Ural State University, Lenina st. 51, 620083, Ekaterinburg, Russia\\
\email{vl.gusev@gmail.com, elena.pribavkina@usu.ru}}

\maketitle

\begin{abstract}
A finite set $S$ of words over the alphabet $\Sigma$ is called non-complete if $\Fact(S^*)\ne\Sigma^*$. A word $w\in\Sigma^*\setminus\Fact(S^*)$ is said to be uncompletable.  We present a series of non-complete sets $S_k$ whose minimal uncompletable words have length $5k^2 - 17k + 13$, where $k\ge 4$ is the maximal length of words in $S_k$. This is an infinite series of counterexamples to Restivo's conjecture, which states that any non-complete set possesses an uncompletable word of length at most $2k^2$.
\end{abstract}

\section{Introduction}
Let $\Sigma$ be a finite alphabet. A finite set $S$ of words over the alphabet $\Sigma$ is called \emph{complete} if $\Fact(S^{*})=\Sigma^*$, i.e. every word over the alphabet $\Sigma$ is a factor of a word of $S^{*}$. If $S$ is not complete, $\Sigma^{*}\setminus \Fact(S^{*})$ is not empty and a word in this set of minimal length is called a \emph{minimal uncompletable word} (with respect to the non-complete set $S$). Its length will be denoted by $uwl(S)$.

The problem of finding minimal uncompletable words and their length was introduced in 1981 by Restivo.
%who conjectured that there is a quadratic upper bound $2k^2$ for the length of a minimal uncompletable word for $S$ in terms of the maximal length $k$ of words in $S$.
In his paper   \cite{Rest} he conjectured that a non-complete set $S$ always possesses an uncompletable word $w$ of length at most $2k^2$, where $k$ is the maximal length of words in $S$, and $w$ is of the form $w=uv_1uv_2\cdots uv_{k-1}u$, where $u\notin S$, $|u|=k$ and $|v_i|\le k$ for all $i=1,2,\ldots,k-1$. An example giving a lower bound $k^2 + k - 1$ for the length of minimal uncompletable words was presented in \cite{Prib09}. However Restivo's conjecture appeared to be false by means of a counterexample found in \cite{FiciPribSakar10}. Namely, let $k>6$ and let $R_k=\Sigma^k\setminus \{a^{k-2}bb\} \cup \Sigma ba^{k-4}\Sigma \cup \Sigma ba \cup b^{4}\cup J_{k}$,
where $J_{k}=\bigcup_{i=1}^{k-3}(ba^{i}\Sigma \cup a^{i}b)$. In \cite{FiciPribSakar10} the authors computed for $7\leq k\leq 12$ that the length of a minimal uncompletable word for $R_k$ is equal to $3k^2 - 9k + 1$ but were unable to prove it in general.

In this paper we present a new series of non-complete sets $S_k$ whose minimal uncompletable
words have length $5k^2 - 17k + 13$ for $k \ge 4$.

As far as the upper bound is concerned, only trivial exponential one is known. More precisely, the length of a minimal uncompletable word is at most $2^{\|S\|-m+1}$, where $m$ is the number of elements in $S$ and $\|S\|$ is the sum of lengths of all elements in $S$. It comes from the connection between non-complete sets and synchronizing automata studied in \cite{Prib09}. However this bound is not likely to be precise.

An interesting related question of deciding whether a given regular language $L$ satisfies one of the properties $\Sigma^*=\Fact(L)$, $\Sigma^*=\Pref(L)$, $\Sigma^*=\Suff(L)$ has been recently considered by Rampersad et al.\ in \cite{RamSha09}, where the computational complexity of the aforementioned problems in case $L$ is represented by a deterministic or non-deterministic finite automaton is studied. In particular case $L=S^{*}$ for $S$ being a finite set of words the authors mention that the complexity of deciding whether or not $\Sigma^*=\Fact(S^{*})$ is still an open problem.

\section{The set $S_k$}

To fix the notation, let us recall some basic definitions from combinatorics on words.
By $|w|$ we denote the \emph{length} of a word $w$. The length of the \emph{empty word} $\varepsilon$ is equal to zero.
By $\Sigma^+$ we denote the set of all non-empty words over the alphabet $\Sigma$;
by $\Sigma^k$ -- the set of all words of length $k$ over $\Sigma$ and by $\Sigma^{\le k}$ -- the set of all words of length at most $k$ over $\Sigma$.
A word $u\in \Sigma^+$ is a \emph{factor} of $w$
(\emph{prefix} or \emph{suffix} respectively) if $w$ can be decomposed as
 $w=xuy$ ($w=uy$ or $w=xu$ respectively)
for some $x,y\in \Sigma^*$. A factor (prefix, suffix) $u$ of
$w$ is called \emph{proper} if $u\ne w.$
Given a word $u=a_1a_2\cdots a_n\in \Sigma^+$ by $u[i\ldots j]$ with
$1\le i,j\le n$ we denote the factor $a_ia_{i+1}\cdots a_{j}$
if $i\le j$, and the empty word if $i>j$. Moreover, we put $u[0]=\varepsilon$.

Let $\Sigma=\{a,b\}$. Consider the set $$S_k=\left(\Sigma^k\setminus\{ba^{k-1},b^{k-1}a\}\right)\cup\left(\Sigma^{k-1}\setminus\{a^{k-1},b^{k-1}\}\right).$$
In section~\ref{upper} we show that this set is not complete for $k\ge4$ and possesses an uncompletable word of length $5k^2-17k+13$. In section~\ref{lower} we show that this upper bound is precise. Our results considerably rely upon the notion of a \emph{forbidden position} in a word. This notion was introduced in \cite{Prib09}. Let $S$ be any non-complete set and let $w$ be an uncompletable word for the set $S$. We say that $0\le j\le |w|-1$ is a forbidden position in $w$ with respect to $S$, if $w[j+1,\ldots,|w|]\notin\Pref(S^*)$, i.e. the suffix of the word $w$ starting from position $j$ is not a prefix of any word in $S^*$. If the set $S$ is clear from context we will omit reference to $S$. Note that, if $S\subseteq\Sigma^{\le k}$ and positions $0,1,\ldots,k-1$ are forbidden in some word $w$, then $w$ is uncompletable for $S$. So to prove that a set $S\subseteq\Sigma^{\le k}$ is not complete, it is enough to find a word with first $k$ forbidden positions.

%It is easy to see from the structure of $S_k$ that any word of length at most $k$ can be covered. So any uncompletable word has length at least $k+1$.

\begin{lemma}\label{FP_struct} Let $j\ge k$ be a forbidden position in a word $w$ with respect to $S_k$. Then the position $j-k$ is forbidden in $w$ with respect to $S_k$ iff either $j-1$ is forbidden or $w[j-k+1,\ldots,j-1]\in\{a^{k-1},b^{k-1}\}$.
\end{lemma}

\begin{proof}  Let $j-k$ be forbidden in $w$. Then by definition $w[j-k+1,\ldots,|w|]\notin\Pref(S_k^*)$. Suppose $j-1$ is not forbidden in $w$, i.e.  $w[j,\ldots,|w|]=x\in\Pref(S^*_k)$. If the factor $y=w[j-k+1,\ldots,j-1]$ of length $k-1$ is in $S_k$, then $w[j-k+1,\ldots,|w|]=yx\in\Pref(S_k^*)$, which is a contradiction. Hence $y\in\Sigma^{k-1}\setminus S_k=\{a^{k-1},b^{k-1}\}$.

Conversely, arguing by contradiction suppose $j-k$ is not forbidden. Then since the length of the suffix $w[j-k+1,\ldots,|w|]$ is at least $k$, it can be factorized as  $xy$ where $x\in S_k$ and $y\in\Pref(S^*_k)$. The case $|x|=k$ contradicts the condition that $j$ is forbidden, since we get $w[j+1,\ldots,|w|]=y\in\Pref(S^*_k)$. Hence $|x|=k-1$ and since $x\in S_k$ we have that $x$ is different both from $a^{k-1}$ and $b^{k-1}$. But then position $j-1$ is not forbidden, which is a contradiction.

%\begin{figure}[ht]
%\begin{center}
%\label{fig_j}
%  \unitlength=1pt
% \begin{picture}(120,15)(0,5)
% \put(0,0){\line(1,0){120}}
% \put(50,-10){$w$}
% \put(20,0){\line(0,1){5}}
% \put(90,0){\line(0,1){5}}
% \put(100,0){\line(0,1){5}}
% \put(60,25){$k$}
% \put(50,8){$k-1$}
% \put(0,0){\line(0,1){5}}
% \put(120,0){\line(0,1){5}}
% %\put(20,7){\oval(40,30)[t]}
% \put(60,7){\oval(80,30)[t]}
% \put(55,7){\oval(70,20)[t]}
% %\put(140,7){\oval(40,30)[t]}
% \end{picture}
% \end{center}
%  \caption{The word $w$ does not contain $u$ as a factor.}
%\end{figure}

\end{proof}

In the rest of the paper we strictly fix the following notation:  $u=ba^{k-1}$ and $v=b^{k-1}a$. We will consider forbidden positions only in occurrences of $u$ and $v$ in $w$. In each such occurrence for convenience we will enumerate forbidden positions locally from $0$ to $k-1$.

\begin{example}
Consider the following word: $$w=\underbrace{'b'a'a}\limits_{\{0,1,2\}}\, a'b\, \lefteqn{\underbrace{\phantom{'bb'a}}\limits_{\{0,2\}}}'b\overbrace{b'a a}\limits^{\{1\}}a\, \underbrace{'bba}\limits_{\{0\}}.$$

Using definition and lemma~\ref{FP_struct} it is easy to calculate the set of its forbidden positions with respect to $S_3$: $\{0,1,2,4,5,7,10\}$.
There are two occurrences of $u$ and two occurrences of $v$ in $w$ (the first occurrence of $v$ overlaps with the second occurrence of $u$).
Locally enumerated sets of forbidden positions are: $\{0,1,2\}$ in the first occurrence of $u$, $\{1\}$ in the second occurrence of $u$, $\{0,2\}$ in the first occurrence of $v$, and $\{0\}$ in the second occurrence of $v$. Note that, since first three positions are forbidden in $w$, this word is uncompletable for $S_3$.
\end{example}

Position $0$ may be forbidden in an occurrence of $u$ or $v$. The following statement gives necessary and sufficient conditions for this to happen. It is an easy consequence of lemma~\ref{FP_struct}.

\begin{lemma}
\label{FP_zero}
Position $0$ is forbidden in an occurrence of $u$ in a word $w$ iff position $k - 1$ is forbidden in the same occurrence of $u$.
Position $0$ is forbidden in any occurrence of $v$ in $w$.
\end{lemma}

Two occurrences $p,q\in \{u,v\}$ in a word $w$ are said to be \emph{consecutive} if they either overlap or are the only occurrences from $\{u,v\}$ in the factor $pxq$ of $w$.

\begin{lemma}
\label{FP_tranform}
Let $p,q\in \{u,v\}$ be two consecutive occurrences without overlap in a word $w$, and let $pxq$ be a factor of $w$ with $|x|\ge0$. Let $F_p$ and $F_q$ be the sets of forbidden positions in $p$ and $q$ respectively. Then $$F_{p}\subseteq\{j+|x|\mod k\mid j\in F_q\}\cup\{0\}.$$
\end{lemma}

\begin{proof}
Consider a forbidden position $i$ in $p$ such that $i\ne0$ and consider the factor $y$ of length multiple of $k$ in $w$ from position $i$ in $p$ to some position $j$ in $q$ ($0\le j<k$). This factor is in $S_k^+$, since $p$ and $q$ are consecutive occurrences of words from $\Sigma^k \setminus S_k$.  Thus, if position $j\notin F_q$, then neither position $i$ is forbidden. On the one hand, we have $|y|\equiv0\mod k$, on the other hand $|y|=k-i+|x|+j$, hence $i\equiv j+|x|\mod k$.
\end{proof}

Note that, two words from $\Sigma^k \setminus S_k$ overlap only in case of $v$ and $u$. More precisely, two last letters of $v$ overlap with first two letter of $u$ leading to the word $b^{k - 1}a^{k - 1}$. The following statement can be easily proved using the same argument as in the previous lemma.

\begin{lemma}
\label{FP_transform_overlap}
Let $v$ and $u$ be two consecutive overlapping occurrences in a word $w$, and let $F_v$ and $F_u$ be the corresponding sets of forbidden positions. Then $$F_{v}\subseteq\{j - 2\mod k\mid j\in F_u\}\cup\{0\}.$$
\end{lemma}

Previous lemmas allow us to make the following observation. Let $p, q \in \{u,v\}$ be two consecutive occurrences in $w$. Then forbidden positions in $p$ except $0$ are inherited from forbidden positions in $q$, and position $0$ may appear in $F_p$ according to lemma~\ref{FP_zero}. In our proofs we will trace backwards forbidden positions only in occurrences of words from $\Sigma^k \setminus S_k$ starting from the last one. Besides, the number of forbidden positions in consecutive occurrences increases by at most~$1$.

\section{Upper bound for $uwl(S_k)$}
\label{upper}

In this section we prove that the set $S_k$ is non-complete by presenting an uncompletable word $w$ of length $5k^2-17k+13$ for $k\ge 4$.

\begin{theorem}
\label{non-complete}
For $k\ge4$ the set $S_k$ is not complete and there exists an uncompletable word of length $5k^2-17k+13$.
\end{theorem}

\begin{proof} For clarity by $r$ we denote overlapping occurrences of $v$ and $u$, i.e. $r=b^{k-1}a^{k-1}$. Consider the word $\omega=u\dt\prod\limits_{i=1}^{k-3}(ra^i\dt b^{k-2-i}r\dt)v$, where $\dt$ can be replaced by any letter of the alphabet $\Sigma$. Let us enumerate occurrences of $v$ counting backwards from the last one (the first occurrence of $v$ in this order will have number $0$), and let $F_{v}^i \subseteq \{0, \ldots, k - 1\}$ be the set of forbidden positions in the $i$th occurrence of $v$. Occurrences of $u$ are counted in the same way (but starting from $1$ instead of $0$) and $F_u^i$ are defined analogously. Note that, $F_v^0 = \{0\}$. Then using previous lemmas it is easy to see that $F_u^1=\{1\}$ and $F_v^1=\{0,k-1\}$. By lemma~\ref{FP_tranform} we have $F_u^2\subseteq\{0,k-2,k-1\}$ and by lemmas~\ref{FP_struct} and \ref{FP_zero} we get $F_u^2=\{0,k-2,k-1\}$. Analogously we obtain $F_v^2=\{0,k-3,k-2\}$.
Now assume $F_v^{2i}=\{0,k-i-2,\ldots,k-2\}$, $1\le i\le k-4$, and let us show that $F_v^{2(i+1)}=\{0,k-i-3,\ldots,k-2\}$. Applying step by step lemmas~\ref{FP_struct}-\ref{FP_transform_overlap} we obtain the following sets of forbidden positions: $$\begin{array}{l}
                        F_u^{2i+1}=\{0,1,k-i,\ldots,k-1\},\\
                        F_v^{2i+1}=\{0,k-i-1,\ldots,k-1\}, \\
                        F_u^{2i+2}=\{0,k-i-2,\ldots,k-1\},\\
                        F_v^{2i+2}=\{0,k-i-3,\ldots,k-2\}.
                       \end{array}
$$
Thus, for $i=k-3$ we have $F_v^{2(k-3)}=\{0,1,\ldots,k-2\}$. Hence the set of forbidden positions in the last occurrence of $u$ is $F_u^{2(k-3)+1}=\{0,1,\ldots,k-1\}$, which means that the word $\omega$ is uncompletable. Its length equals $5k^2-17k+13$.
\end{proof}

\section{Lower bound for $uwl(S_k)$}
\label{lower}

First we prove some nice properties of a minimal uncompletable word in $S_k$.

\begin{theorem}
\label{minimal_border}
Consider a minimal uncompletable word $w$. Then $u$ is a prefix of $w$ and $v$ is its suffix.
\end{theorem}

\begin{proof}
The word $w$ has either $u$ or $v$ as a factor, otherwise $w \in \Pref((S_k \cap \Sigma^{k})^*)$.
Let $w = w'x$, where $|x| = k$, and suppose $x \ne v$. Let $x = x'z$, where $z \in \Sigma$. Since $w$ is minimal, we conclude that $w'x' \in \Fact(S_k^*)$, which means $rw'x' = qy$, for some $q \in S_k^*$ and $|y| \leq k - 1$. If $|y| < k - 1$, then $yz \in \Pref(S_k)$, because all the words of length at most $k - 1$ are prefixes of some words in $S_k$. If $|y| = k - 1$, then $yz = x$. If $x \ne u$, then $yz \in S_k$, and $qyz = rw \in S_k^*$.  If $x = u$, then $y = ba^{k - 2}$, $z = a$ and we have for instance $rwa^{k - 1} \in S_k^*$. In any case we get a contradiction with the fact that $w$ is uncompletable. Thus, $w$ has $v$ as a suffix.

Now we are going to investigate one particular symmetry property of uncompletable words for $S_k$.
It is trivial that the mirror image $\overleftarrow{S}$ of a non-complete set $S$ is again non-complete. Moreover,  mirror images $\overleftarrow{w}$ of uncompletable words $w$ for $S$ are uncompletable for $\overleftarrow{S}$. The same property holds true for renaming morphism: $\varphi(a) = b$ and $\varphi(b) = a$. Applying these statements to our set we get $T_k = \overleftarrow{S_k} = \varphi(S_k)$, where
$$T_k=\left(\Sigma^k\setminus\{a^{k-1}b,ab^{k-1}\}\right)\cup\left(\Sigma^{k-1}\setminus\{a^{k-1},b^{k-1}\}\right).$$ So, if $w$ is an uncompletable word for $S_k$, then $\varphi(\overleftarrow{w})$ is also uncompletable for $S_k$. As we have already shown, every minimal uncompletable word has $v$ as a suffix. From the symmetry property it follows that every such word has $u$ as a prefix.
\end{proof}

The suffix $v$ of a minimal uncompletable word $w$ has only one forbidden position, namely $0$, and in the prefix $u$ of $w$ all the positions from $0$ to $k - 1$ are forbidden.  Thus, we have to analyze how forbidden positions change from one occurrence of a word from $\Sigma^k \setminus S_k$ to the next one.

Consider an arbitrary occurrence of a word from $\Sigma^k \setminus S_k$ in $w$. Let $F$ be the set of its forbidden positions. We will make use of the following representation of $F$: $$F = [f_{1,1}, f_{1, 2}, \ldots, f_{1, m_1}; f_{2,1}, \ldots, f_{2,m_2}; \ldots; f_{n, 1}, \ldots, f_{n, m_n}],$$ where $f_{i, j + 1} \equiv f_{i,j} + 1 \mod k$, $f_{i + 1, 1} > f_{i, m_i}$ and $n\ge 1$. Simply speaking, we partition the set $F$ into blocks of consecutive (with respect to cyclic order) forbidden positions.

\begin{example}
Consider the word $ba^{10}$, and let $F = \{0,1,2,5,6,8,10\}$ be the set of its forbidden positions with respect to $S_{11}$. Then according to our  representation $F = [5,6; 8; 10,0,1,2]$.
\end{example}

\begin{theorem}
\label{increase}
Let $p$ and $q$ be two consecutive occurrences of words from $\{u,v\}$ in a minimal uncompletable word $w$. Let $F_p$ and $F_q$ be the sets of forbidden positions in $p$ and $q$ respectively. If $|F_p| > |F_q|$, then one of the following holds true:
\begin{enumerate}
\item[(i)] $p = q = u$, $F_p = \{0, k - j, \ldots, k - 1\}$ and $F_q = \{1, \ldots, j\}$, where $1 \leq j \leq k - 2$;
\item[(ii)] $p = v$, $q = u$, these occurrences overlap, $F_p=\{0,k-1\}$ and $F_q = \{1\}$;
\item[(iii)] $p = u$, $q = v$, $F_p = \{0, j - i - 1, \ldots, k - 1\}$, $F_q = \{0, \ldots, i, j, \ldots, k - 1\}$, where $j \not\equiv i + 1\mod k$ and there are $k - 1 - i \mod k$ letters between these occurrences;
\end{enumerate}
\end{theorem}

\begin{proof}
Let $F_q = [f_{1,1},  \ldots, f_{1, m_1}; \ldots; f_{n, 1}, \ldots, f_{n, m_n}]$. First assume that $p$ and $q$ do not overlap, so let $pxq$ be the corresponding factor of $w$.

\emph{Case 1.}
Let $p = q = u$. Consider the case $n \geq 2$. Then there exists $1 \leq i \leq n$ such that $f_{i, 1} \geq 2$. If $f_{i,1} > 2$, then by lemma~\ref{FP_struct} position $f_{i,1} + |x| \notin F_p$. If $f_{i,1} = 2$, then by the same lemma $f_{i,1} + |x| \in F_p$ implies $x = x'b^{k - 2}$. But then $v$ is a factor of $pxq$ which contradicts the fact that $p$ and $q$ are two consecutive occurrences from $\{u,v\}$. Thus, $f_{i,1} + |x|$ is not forbidden in $p$. From lemma~\ref{FP_tranform} it follows that $F_p \subseteq\{j+|x|\mod k\mid j\in F_q\}\cup\{0\}$, whence $|F_p| \leq |F_q|$, a contradiction. Consequently $n = 1$ and $F_q = [f_{1,1}, \ldots, f_{1, m}]$. If $f_{1,1} \geq 2$, then following the same argument as above, we conclude that $|F_p| \leq |F_q|$, hence $f_{1,1} \in \{0,1\}$. Note that, if $f_{1,1} = 0$, by lemma~\ref{FP_zero} we have $k - 1 \in F_q$. It means that $F_q = \{0,1, \ldots, k - 1\}$ which contradicts minimality of $w$. Thus $f_{1,1} = 1$. If $f_{1,m} = k - 1$, then by lemma~\ref{FP_zero} we obtain that $0 \in F_q$ and in this case $f_{1,1}$ cannot be equal to $1$. So $f_{1,m} \leq k - 2$.
Now it remains to prove that $F_p$ has form as stated in $(i)$. %By lemma ~\ref{FP_struct} position $1 + |x| \in F_p$ only if $x = a^{\ell}$ for some $\ell$. 
Since $|F_p| > |F_q|$, we have $0 \in F_p$. Then by lemma~\ref{FP_zero} we obtain that $k - 1$ is also in $F_p$. So there exists a position $i \in F_q$ satisfying $i + |x| \equiv k - 1 \mod k$. If $i < f_{1,m}$, then $i + 1 + |x| \equiv 0 \mod k$, hence $|F_p| \leq |F_q|$. Thus $i = f_{1,m}$, $|x| \equiv k - f_{1,m} - 1 \mod k$ and by lemma~\ref{FP_tranform} we deduce $F_p = \{0, k - f_{1,m}, \ldots, k - 1\}$.

\emph{Case 2.}
Let $p = v$, $q = u$. We are to show that either $f_{1,1} + |x|$ is not forbidden in $p$ or $f_{1,1} + |x| \equiv 0 \mod k$, so in both cases $|F_p| \leq |F_q|$. It holds for $f_{1,1} \geq 2$ by the argument as in the previous case. Note that, $f_{1,1} \ne 0$, otherwise by lemma~\ref{FP_zero} position $k - 1 \in F_q$ which contradicts our representation of a set of forbidden positions. Therefore, assume $f_{1,1} = 1$. Suppose $1 + |x| \equiv i \mod k$. If $0 < i < k - 1$, then by lemma~\ref{FP_struct} position $i$ is not forbidden in $p$. If $i = 0$, then lemma~\ref{FP_tranform} implies $|F_p| \leq |F_q|$. If $i = k - 1$, then by lemma~\ref{FP_struct} we have $i \in F_p$ only if $a^{k - 2}$ is a prefix of $x$, but then $p$ and $q$ are not consecutive occurrences from $\{u,v\}$.

\emph{Case 3.}
Let $p = u$, $q = v$. Suppose $n\ge 2$, and consider arbitrary positions $f_{i, 1} < f_{j, 1}\in F_q$. We show that either $f_{i, 1}+|x|\mod k\notin F_p$ or $f_{j,1}+|x|\mod k\notin F_p$. Arguing by contradiction, suppose both positions are forbidden in $p$. Then by lemma~\ref{FP_struct} the word $x$ must have suffix $b^{k-f_{j, 1}}$ (if $|x|< k-f_{j,1}$, then $f_{j,1}+|x|\mod k\notin F_p$ by lemma~\ref{FP_struct}). Analogously $b^{k-f_{i,1}}$ must be the suffix of $x$, and if $|x|< k-f_{i,1}$, then $f_{i,1}+|x|\mod k \notin F_p$ by lemma~\ref{FP_struct}. But then by the same lemma $k+(k-f_{j,1})$ last letters of $x$ are $b$'s. Continuing this argument we get that $\ell k - f_{j,1}$ last letters of $x$ are $b$'s for any positive integer $\ell$. It means that there exists no finite word $x$ such that both positions $f_{i,1}+|x|\mod k$ and $f_{j,1}+|x|\mod k$ are forbidden in $p$. Hence $n=1$, and $F_q = [f_{1,1}, \ldots, f_{1,m}]$. Since $|F_p| > |F_q|$, by lemmas~\ref{FP_tranform} and~\ref{FP_zero} there is no $j$ such that $f_{1,j} + |x| \equiv 0 \mod k$. Besides $f_{1, i} + |x| \equiv k - 1 \mod k$ for some $i$. If $i < m$, then $f_{1, i + 1} + |x| \equiv 0 \mod k$, which is impossible. Thus, $i = m$. Moreover, $|x| \equiv k - 1 - f_{1,m} \mod k$. Since $0$ is always forbidden in $v$, we can represent $F_q$ as $\{0, \ldots i, j, \ldots, k - 1\}$ for some $0 \leq i < j \leq k - 1$. Then by lemma~\ref{FP_tranform} we have $F_p = \{0, j - i - 1, \ldots, k - 1\}$.

\emph{Case 4.}
Let $p = v$, $q = v$. If $f_{1,1} + |x| \not\in F_p$ then $|F_p| \leq |F_q|$. Otherwise by lemma~\ref{FP_struct} either $f_{1,1} + |x| \equiv 0 \mod k$ or $f_{1,1} + |x| \equiv k - 1 \mod k$. In the first case by lemma~\ref{FP_tranform} we get $|F_p| \leq |F_q|$. In the latter case by lemma~\ref{FP_struct} the word $a^{k - 2}$ have to be the prefix of $x$, which contradicts the fact that $p$ and $q$ are consecutive occurrences.

Now assume that $p = v$, $q = u$, and they overlap. If there exists $i$ such that $f_{i,1} > 2$ then by lemmas~\ref{FP_struct} and \ref{FP_transform_overlap} position $f_{i,1} - 2 \not\in F_p$, therefore $|F_p| \leq |F_q|$. Note that, $f_{1,1} \ne 0$, otherwise by lemma~\ref{FP_zero} position $k - 1 \in F_q$ which contradicts our representation of a set of forbidden positions. Hence, for all $i$ we have $f_{i, 1} \in \{1,2\}$. It immediately implies that $n = 1$. If $f_{1,j} = 2$ for some $j$ then by lemma~\ref{FP_transform_overlap} we get that $f_{1,j} - 2$ is equal to $0$ and $|F_p| \leq |F_q|$. Thus $F_q = \{1\}$ and $F_p = \{0, k - 1\}$.
\end{proof}

\begin{lemma}
\label{uu=}
Let $p$ and $q$ be two consecutive occurrences of $u$ in a word $w$, $F_p$ and $F_q$ be the corresponding sets of forbidden positions. If $F_q = \{0,i,\ldots, k - 1\}$ and $|F_p| = |F_q|$, then $F_p = F_q$.
\end{lemma}

\begin{proof}
If $i=1$, then the statement of lemma obviously holds true. So we may assume $i>1$. Let $pxq$ be the factor of $w$. Then by lemma~\ref{FP_tranform} we have $F_p\subseteq\{0,|x|\mod k, i+|x|\mod k, i+1+|x|\mod k,\ldots,k-1+|x|\mod k\}$. By lemma~\ref{FP_struct} position $i+|x|\mod k\notin F_p$. Thus in order to have $|F_p|=|F_q|$, it is necessary that $0\in F_p$ and $j+|x|\not\equiv 0\mod k$ for all $i+1\le j\le k$. Lemma~\ref{FP_zero} implies $k-1\in F_p$. It means that $j+|x|\equiv k-1\mod k$ for some $i+1\le j\le k$. It is straightforward that $j=k$, $|x|\equiv k-1\mod k$ and $F_p=\{0,i,i+1,\ldots,k-1\}=F_q$.
\end{proof}

\begin{lemma}
\label{vu=}
Let $p$ and $q$ be consecutive occurrences of $v$ and $u$ respectively in a minimal uncompletable word $w$. Let $F_p$ and $F_q$ be the corresponding sets of forbidden positions. If $F_q = \{0,i,\ldots, k - 1\}$ and $|F_p| = |F_q|$, then these occurrences overlap and $F_p = \{0, i - 1, \ldots, k - 2\}$.
\end{lemma}

\begin{proof}
If $i=1$, then trivially $w$ is not a minimal uncompletable word, so $i > 1$.
Suppose $p$ and $q$ do not overlap, so there exists $x$ such that $pxq$ is a factor of $w$. By lemma~\ref{FP_tranform} we have $F_p\subseteq\{0,|x|\mod k, i+|x|\mod k, i+1+|x|\mod k,\ldots,k-1+|x|\mod k\}$ and from lemma~\ref{FP_struct} it follows that either $i+|x|\notin F_p$ or $i+|x|\equiv 0\mod k$. Moreover, if $j+|x|\equiv 0\mod k$ for some $i+1\le j\le k$, then we immediately get $|F_p|<|F_q|$. So we may assume that $j+|x|\not\equiv 0\mod k$ for all $i+1\le j\le k$. First let $|x|>k-(i+1)$. Then by lemma~\ref{FP_struct} if position $i+1+|x|\mod k$ is forbidden in $p$, then either $i+1+|x|\equiv 0\mod k$ or $i+1+|x|\equiv k-1\mod k$. The first case contradicts our assumption. In the latter case $x = a^{k - 2}x'$, but this contradicts the fact that $p$ and $q$ are consecutive occurrences. Thus, both cases are impossible. So $0\le |x|\le k-(i+1)$, but then we have $i+1\le j=k-|x|\le k$. It means that $j\mod k$ is a forbidden position in $q$ and $j+|x|\equiv 0\mod k$, which again contradicts our assumption that $j+|x|\not\equiv 0\mod k$ for all $i+1\le j\le k$.
Therefore occurrences $p$ and $q$ overlap. By lemmas~\ref{FP_struct} and \ref{FP_transform_overlap} we get $F_{p}=\{0, i-1,\ldots, k-2\}$.
\end{proof}

%\begin{lemma}
%\label{uu-}
%Let $p$ and $q$ be two consecutive occurrences of $u$ in $w$, $F_p$ and $F_q$ be the corresponding sets of forbidden positions. If $F_p = \{1,\ldots, j\}$ then $|F_p| < |F_q|$.
%\end{lemma}

%\begin{lemma}
%Consider consecutive occurrences of $v$ and $u$ in a minimal uncompletable word $w$. Let $F_v$ and $F_u$ be the corresponding sets of forbidden positions. If $F_u = \{f_{1,1},  \ldots, f_{1, m_1}; \ldots; f_{n, 1}, \ldots, f_{n, m_n}\}$ and $n > 1$ then $|F_v| < |F_u|$.
%\end{lemma}

\begin{lemma}
\label{vv-}
Let $p$ and $q$ be two consecutive occurrences of $v$ in a minimal uncompletable word $w$, and let $F_p$ and $F_q$ be the corresponding sets of forbidden positions. If $F_q = \{0,i,\ldots, k - 2\}$ then $|F_p| < |F_q|$.
\end{lemma}

\begin{proof}
Consider the factor $pxq$ of $w$. By lemma~\ref{FP_tranform} we get $F_{p}\subseteq\{0,|x|\mod k,i+|x|\mod k, i+1+|x|\mod k,\ldots, k-2+|x|\mod k\}$. We are going to show that $|x| \mod k$ is not forbidden in $p$. From lemma~\ref{FP_struct} it trivially follows that the position $|x| \mod k \in F_p$ if either $|x| \equiv 0 \mod k$ or $|x| \equiv k - 1 \mod k$. The first case contradicts minimality of $w$, for we would have $F_p=F_q$. In the latter case by the same lemma we conclude $x = a^{k - 2}x'$, which contradicts the fact that $p$ and $q$ are consecutive occurrences. Similar arguments can be applied to $i + |x| \mod k$. Namely, if $i + |x| \equiv k - 1 \mod k$, then again $x = a^{k - 2}x'$, a contradiction. So either $i + |x| \equiv 0 \mod k$ or $i+|x|\mod k\notin F_p$. In both cases we have $|F_p| < |F_q|$.
\end{proof}

\begin{lemma}
\label{uv_special}
Let  $p$ and $q$ be consecutive occurrences of $u$ and $v$ respectively in a word $w$. Let $F_p$ and $F_q$ be the corresponding sets of forbidden positions, and let $|F_p| = |F_q|$. If $F_q = \{0, i, \ldots, k - 2\}$ with $i > 1$ and $F_p=[f_{1,1},\ldots,f_{1,m}]$, %has one block 
then either $F_p = \{0,1, i + 2, i + 3, \ldots, k - 1\}$ or $F_p = \{0, i + 1, i + 2, \ldots, k - 1\}$.
\end{lemma}

\begin{proof}
Consider the factor $pxq$ of $w$. By lemma~\ref{FP_tranform} we get $F_{p}\subseteq\{0,|x|\mod k,i+|x|\mod k, i+1+|x|\mod k,\ldots, k-2+|x|\mod k\}$. Arguing as in  case~3 of the theorem~\ref{increase} we easily obtain that either $|x| \mod k \notin F_p$ or $i + |x| \mod k \notin F_p$. Suppose that  position $|x| \mod k$ is not forbidden in $F_p$. Since $|F_p| = |F_q|$ it is necessary that $0 \in F_p$, therefore  $k - 1$ is also in $F_p$ by lemma~\ref{FP_zero}. Every $i \leq j \leq k - 3$ satisfies $j + |x| \not\equiv k - 1 \mod k$, since otherwise we would have $j + 1 + |x| \equiv 0 \mod k$, which would imply $|F_p| < |F_q|$. Therefore, $k - 2 + |x| \equiv k - 1\mod k$, whence $|x| \equiv 1 \mod k$ and $F_p = \{0, i + 1, \ldots, k - 1\}$. Assume now that position $ i + |x| \mod k \notin F_p$. Following the same argument as in the previous case, we get $k - 1 \in F_p$ and either $|x| \equiv k - 1 \mod k$ or $k - 2 + |x| \equiv k - 1 \mod k$. In the first case $F_p = \{0, i, \ldots, k - 3, k - 1\}$, and it has more than one block of consecutive forbidden positions, which contradicts the statement of lemma. In the latter case $|x| \equiv 1 \mod k$ and $F_p = \{0, 1, i+2, \ldots, k-1\}$.
\end{proof}

\begin{theorem}
The length of a minimal uncomletable word for $S_k$ is at least $5k^2 - 17k + 13$ for $k \ge 4$.
\end{theorem}

\begin{proof}
Let $w$ be an arbitrary minimal uncomletable word for $S_k$. By theorem~\ref{minimal_border} the word $u$ is a prefix of $w$ and $v$ is its suffix. Note that, we have $F_{u}=\{0,1,\ldots,k-1\}$ and $F_v=\{0\}$ in the aforementioned occurrences of $u$ and $v$.

If $s$ and $t$ are two consecutive occurrences from $\{u,v\}$ such that $|F_s| > |F_t|$, then we say that $s$ is an \emph{increasing occurrence}. Recall that by lemma~\ref{FP_tranform} it means that $|F_s| = |F_t| + 1$.  Since there is only one forbidden position in the last occurrence of $v$ in $w$ and $k$ forbidden positions in the first occurrence of $u$, there must be at least $k-1$ increasing occurrences in $w$. Now we are going to estimate the length of a factor between two consecutive such occurrences.
Consider a factor $pxq$ of $w$ such that $p$ and $q$ are the only increasing occurrences inside this factor. Note that, for an occurrence $r \in \{u,v\}$ in $pxq$, different from $p$ and $q$ (if any), $|F_r| = |F_q|$. Otherwise $p$ and $q$ are not consecutive increasing occurrences.

Let $3 \leq |F_q| < |F_p| \leq k$. Then by theorem~\ref{increase} we have $p = q = u$, $F_q = \{0, i, \ldots, k - 1\}$ and $F_p = \{0, i - 1, \ldots, k - 1\}$. Moreover, $p$ and $q$ are not the only occurrences from $\{u,v\}$ in $pxq$. Assume first $i> 2$, i.e. $|F_p|<k$.

Suppose that there is only one occurrence $r \in \{u,v\}$ in $pxq$ different from $p$ and $q$. Then from lemma~\ref{uu=} it follows that $r = v$ and since $|F_r| = |F_q|$, applying lemma~\ref{vu=} we obtain $F_r = \{0, i - 1, \ldots, k - 2\}$. But then since $i>2$ the set $F_r$ does not have the form required by theorem~\ref{increase} for the condition $|F_p| > |F_r|$ to hold true.

Assume now that $r_1,r_2 \in \{u,v\}$ are the two occurrences in $pxq$ different from $p$ and $q$. By the same argument as above $r_2 = v$ and $F_{r_2} = \{0, i - 1, \ldots, k - 2\}$. If $r_1 = u$, then on the one hand by theorem~\ref{increase} it should be $F_{r_1} = \{1, 2, \ldots, k - i + 1\}$. On the other hand by lemma~\ref{uv_special} position $0$ have to be forbidden in $F_{r_1}$, a contradiction. If $r_2 = v$, then by lemma~\ref{vv-} we have $|F_{r_1}| < |F_{r_2}|$, which is impossible. So there are at least three occurrences of words from $\{u,v\}$ in $pxq$ except $p$ and $q$.

Suppose that $r_1, r_2, r_3 \in \{u,v\}$. As we have already seen before $r_3 = v$ and $F_{r_3} = \{0, i - 1, \ldots, k - 2\}$. It immediately follows from lemma~\ref{vv-} that $r_2 = u$. Then $F_{r_2}$ is either $\{0,1, i + 1, i + 2, \ldots, k - 1\}$ or $\{0, i, i + 1, \ldots, k - 1\}$ by lemma~\ref{uv_special}. The latter case contradicts minimality of $w$, since we would have $|F_{r_2}|=|F_{q}|$.
Assume $r_1=u$, then by theorem~\ref{increase} we have $F_{r_1}=\{1,2,\ldots,k-i+1\}$. Let $r_1yr_2$ be a factor of $pxq$. Note that, $0\notin F_{r_1}$ and by lemma~\ref{FP_struct} position $i+1+|y|\mod k\notin F_{r_1}$. Thus $|F_{r_1}|<|F_{r_2}|$, a contradiction. Therefore $r_1=v$. This case is possible and exactly this situation takes place in the word presented in theorem~\ref{non-complete}.

Now we are going to estimate the length of $x$ in this case. By lemma~\ref{vu=} the factor $r_3$ overlaps with $q$, so we can factorize $x$ in the following way: either $x=y_1r_1y_2r_2y_3b^{k-2}$ or, if $r_1$ and $r_2$ overlap, $x=z_1b^{k-1}a^{k-1}z_2b^{k-2}$. By the proof of lemma~\ref{uv_special}, since $F_{r_2}=\{0,1, i + 1, i + 2, \ldots, k - 1\}$, we have $|y_3|\equiv |z_2|\equiv 1\mod k$. Let us first estimate the length of $z_1$. Since $r_1$ and $r_2$ overlap, lemma~\ref{FP_transform_overlap} implies $F_{r_1}=\{0,i,i+1\ldots,k-1\}$. Then by theorem~\ref{increase} there are $k-1\mod k$ letters between $p$ and $r_1$, thus $|z_1|\equiv k-1\mod k$. So in this case we have $|x|\ge 4k-4$. Now let us assume that $r_1$ and $r_2$ are not overlapping. By theorem~\ref{increase} the factor $r_1$ must have only one block of consecutive forbidden positions. By lemma~\ref{FP_tranform} we have $F_{r_1}\subseteq\{0, i+1+|y_2|\mod k, i+2+|y_2|\mod k,\ldots, k-1+|y_2|\mod k, |y_2|\mod k, |y_2|+1\mod k\}$. Note that, by lemma~\ref{FP_struct} either $i+1+|y_2|\mod k\notin F_{r_1}$ or $i+1+|y_2|\equiv 0\mod k$. So for the set $F_{r_1}$ to have only one block of consecutive forbidden positions and the same cardinality as $F_{r_2}$ we must have either $1+|y_2|\equiv k-1\mod k$ or $i+2+|y_2|\equiv 1\mod k$. In the first case we have $|y_2|\equiv k-2\mod k$ and $F_{r_1}=\{0,i,i+1,\ldots,k-1\}$, but this is the same as in the case of overlapping occurrences $r_1$ and $r_2$, so this is impossible in a minimal uncompletable word. In the second case we have $|y_2|\equiv k-i-1\mod k$ and $F_{r_1}=\{0,1,2,\ldots, k-i\}$. Then by theorem~\ref{increase} we conclude that $|y_1|\equiv i-1\mod k$. Therefore in this case we have $|x|\ge 4k-3$.

It is not hard to see that, if there are more than $3$ occurrences from $\{u,v\}$ different from $p$ and $q$ in $pxq$, then even if some of them overlap, the total length of the word $x$ is at least $4k-4$. So we conclude that $|x|\ge 4k-4$. Note that, $|x|=4k-4$ for the word from theorem~\ref{non-complete}.

Now let $i=2$, i.e. $F_p=\{0,1,\ldots,k-1\}$ and $F_q=\{0,2,3,\ldots,k-1\}$. Arguing as above, by lemmas~\ref{uu=} and \ref{vu=} there is an occurrence $r=v$ just before $q$, overlapping with $q$ and $F_r=\{0,1,\ldots,k-2\}$. If this is the only occurrence and $pxq=pyb^{k-2}q$, then by theorem~\ref{increase} we have $|y|\equiv 1\mod k$ and $F_p=\{0,1,\ldots,k-1\}$, so here we have $|x|\ge k+1$. More occurrences from $\{u,v\}$ inside $pxq$ will obviously give a longer factor $x$.

The previous argument implies that any minimal uncompletable word has prefix $pyb^{k-2}q$ with $p=q=u$ and $|y|\equiv 1\mod k$. From the symmetry property observed in theorem\ref{minimal_border} we deduce that any minimal uncompletable word has suffix $q'a^{k-2}\hat{y}p'$ with $p'=q'=v$ and $|\hat{y}|\equiv 1\mod k$. Clearly  $F_{p'}=\{0\}$. To calculate $F_{q'}$ note that, there is an occurrence $r=u$ overlapping with $q'$ and $F_r=\{1\}$. From theorem~\ref{increase} we deduce that $q'$ is an increasing occurrence and $F_{q'}=\{0,k-1\}$. Let $p$ be the next increasing occurrence. If $|F_p| = 3$, then the same theorem implies that $p=u$, $F_p=\{0,k-2,k-1\}$ and there are at least $k-1$ letters between $p$ and $q'$. If $|F_p| \leq 2$, then we would obviously require at least $k$ letters between $q'$ and an increasing occurrence with three forbidden positions.

Thus to increase the number of forbidden positions from $1$ in the suffix $v$ of $w$ to $2$ we need at least $k-1$ letters; from $2$ to $3$ -- at least $k-1$; from $\ell$ to $\ell+1$ for $3\le \ell\le k-2$ we need at least $4k-4$ letters, and finally, from $k-1$ to $k$ we need $k-1$ letters. Besides we have at least $k-1$ increasing occurrences and the suffix $v$ of $w$ with only one forbidden position. Thus the length of a minimal uncompletable word is at least $3(k-1)+(4k-4)(k-4)+k(k-1)+k=5k^2-17k+13.$
\end{proof}

\section{Conclusion}
The series of sets $S_k$ was found during exhaustive computational experiment. We searched for maximal with respect to inclusion non-complete sets among all the subsets of $\Sigma^{\leq 3}$; we were interested in such sets having longest possible minimal uncompletable word. We have found two extreme sets up to renaming letters and taking mirror image. Namely, $S_3 = \left(\Sigma^3\setminus\{baa,bba\}\right)\cup\left(\Sigma^{2}\setminus\{aa,bb\}\right)$ and $\left(\Sigma^3\setminus\{baa,bba\}\right)\cup\left(\Sigma^{2}\setminus\{ab,ba\}\right)$. Computation was based on representation of a set $S^*$ as a flower automaton and on the fact that $S$ is non-complete if and only if the corresponding non-deterministic automaton is synchronizing. Moreover, the set of uncompletable words coincides with the language of synchronizing words, see \cite{BePeReu09} and \cite{Prib09}  for more details. The same task for $k = 4$ was unfeasible for a typical laptop, so the search was performed with restriction $|\Sigma^{4} \cap S| \geq 11$. There is only one extreme non-complete set up to renaming letters and taking mirror image in this class. Namely, $\left(\Sigma^4\setminus\{aabb, abaa, abbb\}\right)\cup\left(\Sigma^{3}\setminus\{aba, bba, bbb\}\right)$. The length of a minimal uncompletable word for this example is $31$ compared to $25$ for the set $S_k$. So $S_k$ is not optimal even for $k = 4$. Thus, the lower bound $5k^2 - 17k + 13$ is likely to be improved. Nevertheless, the most interesting question whether the tight bound is quadratic remains open.

\end{document}